\documentclass[%
reprint,
superscriptaddress,
frontmatterverbose, 
preprintnumbers,
 amsmath,amssymb,
 aps,
 pra,
]{revtex4-1}

\usepackage{graphicx}
\usepackage{dcolumn}
\usepackage{bm}

\usepackage{amsthm}
\theoremstyle{plain}
\newtheorem{thm}{Theorem}
\newtheorem{lem}[thm]{Lemma}

\newtheorem{cor}[thm]{Corollary}

\theoremstyle{definition}
\newtheorem{defn}{Definition}

\newcommand{\ket}[1]{|#1\rangle}
\newcommand{\bra}[1]{\langle #1|}
\newcommand{\bracket}[2]{\langle #1|#2\rangle}
\newcommand{\ketbra}[2]{|#1\rangle\langle #2|}

\begin{document}

\preprint{MIT-CTP/4692}

\title{Doubly infinite separation of quantum information and communication}

\author{Zi-Wen Liu}\email{zwliu@mit.edu}
\affiliation{Center for Theoretical Physics and Department of Physics, Massachusetts Institute of Technology, Cambridge, Massachusetts 02139, USA}
\author{Christopher Perry}
\affiliation{Department of Physics and Astronomy, University College London, Gower Street, London WC1E 6BT, United Kingdom}

\author{Yechao Zhu}
\affiliation{Center for Theoretical Physics and Department of Physics, Massachusetts Institute of Technology, Cambridge, Massachusetts 02139, USA}

\author{Dax Enshan Koh}%
\affiliation{%
Department of Mathematics, 
Massachusetts Institute of Technology, Cambridge, Massachusetts 02139, USA
}%

 \author{Scott Aaronson}
\affiliation{Computer Science and Artificial Intelligence Laboratory, 
Massachusetts Institute of Technology, Cambridge, Massachusetts 02139, USA
}

\date{\today}

\begin{abstract}

We prove the existence of (one-way) communication tasks with a subconstant versus superconstant asymptotic gap, which we call ``doubly infinite,'' between their quantum information and communication complexities. We do so by studying the exclusion game [C. Perry {\it et al.}, Phys. Rev. Lett. {\bf 115}, 030504 (2015)] for which there exist instances where the quantum information complexity tends to zero as the size of the input $n$ increases.  By showing that the quantum communication complexity of these games scales at least logarithmically in $n$, we obtain our result. We further show that the established lower bounds and gaps still hold even if we allow a small probability of error. However in this case, the $n$-qubit quantum message of the zero-error strategy can be compressed polynomially.
\end{abstract}

\maketitle


\section{Introduction}
\def\EXC{{\rm EXC}}

The field of communication complexity, originated by Yao's 1979 seminal work \cite{yao}, aims to study the minimum amount of communication needed for multiple distributed parties to accomplish a given communication task. Such tasks are typically formalized as follows: Players are given private inputs and asked to solve some computational problems based on them. To do this, some communication will have to take place in the form of exchanging messages. 

While such models were originally considered in the context of classical protocols, it has since been realized that quantum resources, e.g., quantum communication channels (players are allowed to exchange quantum states instead of classical messages), may offer significant advantage. For example, there exist tasks for which quantum strategies can consume exponentially less communication than any classical one, even without shared entanglement \cite{Raz,fp,yjk,bcw,gkk,Montanaro}.
These results sparked interest in further characterizing which tasks exhibit distinctions between quantum and classical communication protocols, and what kind of distinctions there can be. 
The vast majority of previous work in this field was carried out in the constant bounded-error setting. 
Here we shall focus on a scenario where the allowed probability of error is zero or vanishingly small.


Recently, a peculiar type of one-way communication task between two players Alice and Bob, namely the exclusion game, was introduced by Perry, Jain, and Oppenheim (PJO) \cite{comm}: Alice randomly draws an $n$-bit string $x$, and Bob randomly draws some subset $y \subseteq [n]$, where $|y|=m$, both from uniform distributions. Alice then sends a single message regarding her input to Bob. They win the game if Bob outputs a string $z$ that is different from $x$ restricted to the bits specified by $y$. A particular exclusion game can be denoted by $\EXC_{n,m,\gamma}$, where $\gamma$ is the allowed probability of error. Comparing to the conventional bounded-error tasks of computing functions, exclusion games exhibit some special properties. They are relational tasks: multiple outputs can be accepted for one certain input; and they are extremely sensitive to error: if $\gamma \geq 2^{-m}$, then no communication is required as Bob can succeed at his task by guessing a string at random.
In Ref.~\cite{comm}, PJO demonstrated a new kind of quantum-classical separation: They devised a zero-error quantum strategy that only reveals vanishingly small amount of information regarding Alice's input for the exclusion games with large $m$, while any zero-error classical strategy must reveal almost everything. More formally, there is an infinite gap between the quantum and classical information complexities (the minimum amount of information regarding Alice's private input that needs to be revealed) of these exclusion games. 



In this paper, we further analyze the complexities of different exclusion games, and exhibit several features.
The PJO strategy requires that exactly $n$ qubits be sent from Alice to Bob, i.e., the communication cost is $n$. Since the amount of information actually revealed is vanishingly small, an interesting question that naturally arises is as to how much we can possibly reduce the communication cost. 
For zero-error exclusion games with $m$ scaling strongly sublinearly in $n$, we show that any winning quantum strategy can be classically simulated with at most exponential overhead. Combining with the linear lower bound on the classical communication complexity, we establish a logarithmic lower bound on the quantum communication complexity of these games. As a result, there is an at least subconstant versus logarithmic gap between the quantum information and communication complexities of the exclusion games for which both sides hold simultaneously (they exist). That is, a vanishingly small amount of extractable information must be carried by a diverging amount of communication for these tasks. This gap is an example of doubly infinite gaps, which we shall motivate and define later.

Next, we extend our analysis to the cases where error may be allowed ($\gamma>0$). By slightly different arguments, we show that the overhead of classically simulating a quantum strategy is still at most exponential for small $\gamma$. Furthermore,
for $\gamma\leq(n+1)^{-m}$, we show that the classical communication complexity is still at least linear, thus the logarithmic lower bound on the quantum communication complexity and the doubly infinite gap between the quantum information and communication complexities of certain exclusion games hold. 

The significance of the doubly infinite gap between quantum information and communication complexities may be compared with its classical counterpart. For constant non-zero probability of error, the gap between classical information and communication complexities is at most exponential for any communication task \cite{gkr,Braverman}. For zero and asymptotically vanishing probability of error, the largest known gap is constant versus superconstant (``singly infinite'') and occurs for the equality function \cite{Braverman}. Our results may lead to a better understanding of the relation between information and communication complexities, which is a major objective of recent research in the field of communication complexity (in both classical and quantum settings).

We should mention that with regards to the gaps between quantum and classical communication complexities, it was shown in Ref.\ \cite{kremer} that for computing functions in the bounded-error model (without shared entanglement or randomness), they can be at most exponential. As the exclusion game is a relational problem and the interesting separations occur only when the probability or error is zero or tends to zero asymptotically, the arguments of Ref.\ \cite{kremer} cannot be directly applied here.
Our results indicate that the conclusion holds for almost all exclusion games. However, it indeed remains open as to whether the gap can be superexponential for those games with $m$ scaling linearly in $n$.

In addition, we show that $\gamma\geq(n+1)^{-m}$ allows the $n$-qubit quantum communication of the PJO strategy to be compressed at least polynomially. Most of our results are summarized in Fig.~\ref{excf}.
\begin{figure}
\centering
\includegraphics[width=\columnwidth]{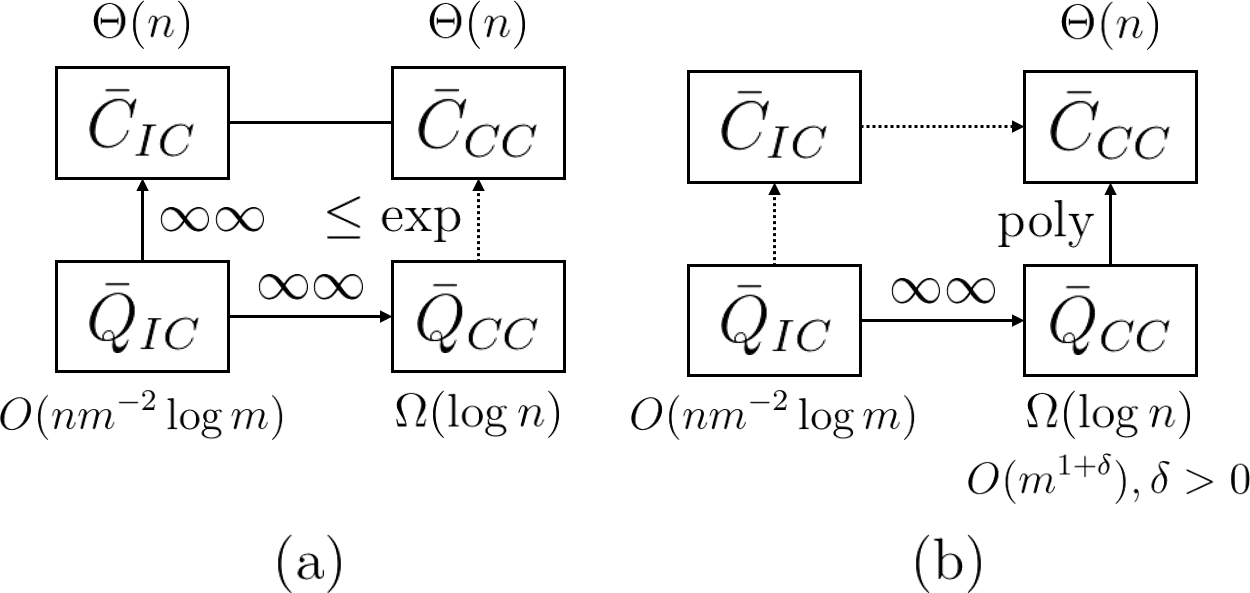}
\caption{Complexities of $\EXC_{n,m,\gamma}$ with (a) $m\in\omega(\sqrt{n\log n})$, $m\in \tilde o(n)$, $\gamma=0$; (b) $m\in\omega(\sqrt{n\log n})$, $m\in \tilde o(n)$, $\gamma=(n+1)^{-m}$. Solid arrows denote established gaps (pointing towards the larger complexity), while the dashed ones denote unknown gaps. ``$\infty\infty$'' means doubly infinite.}
\label{excf}
\end{figure}

\section{Complexities of communication tasks}



Two types of information-theoretic quantities associated with a certain communication task are of great interest and importance in our context, namely, the communication \cite{yao,Kushilevitz} and information \cite{cywy} complexities. Here, we formally define them.


The communication cost of a $\lambda$ protocol $\Pi$ [where $\lambda = C$ (classical) or $=Q$ (quantum) in our context], denoted by $\lambda_{CC}(\Pi)$, is defined to be the maximum number of bits or qubits that are exchanged in any run of the protocol, where the maximum is taken over all inputs and the value of any randomness used.

The information cost of a $\lambda$ protocol $\Pi$, denoted by $\lambda_{IC}(\Pi)$, aims to measure the amount of information regarding the players' inputs revealed by $\Pi$. Here, we consider one-way protocols, i.e., the communication is only from Alice to Bob.
Suppose that $X$ and $Y$ are, respectively, Alice and Bob's inputs, distributed according to a joint distribution $\mu$. Then
$\lambda^{\mu}_{IC}\left(\Pi\right)= I\left(X:\Pi|Y\right)$, where $\Pi$ on the right-hand side essentially denotes the message exchanged during the protocol together with the public randomness used, and $I(S:T|U) = H(SU)+H(TU)-H(STU)-H(U)$ measures the mutual information between $S$ and $T$ given knowledge of $U$ \cite{ic}. The distribution-independent information cost is then defined to be $\lambda_{IC}(\Pi)=\sup_\mu \lambda^{\mu}_{IC}(\Pi)$ \cite{Braverman}. 

Complexities measure the minimum possible amount of certain costs that need to take place to accomplish the task, where the minimization is taken over all winning protocols. The $\lambda$ communication complexity of a task $\Xi$ is then defined to be $\bar\lambda_{CC}(\Xi) = \inf_{\Pi_\Xi} \lambda_{CC}(\Pi_\Xi)$, where $\Pi_\Xi$ are all winning $\lambda$ protocols for $\Xi$. The distribution-dependent and distribution-independent $\lambda$ information complexities of one-way tasks are defined similarly.

We emphasize that these quantities of interest are only associated with the communication between players. Players themselves can have unlimited access to any kind of local resources.

\section{Infinite asymptotic gaps}

We are interested in the limiting behaviors of complexities as the size of the task $n$ tends to infinity.
Throughout this paper, we adopt the standard notation to describe asymptotic complexities. In addition to the widely used $O,o,\Omega,\omega$ (Bachmann-Landau) symbols (formal definitions can be found in, e.g., Ref.\ \cite{knuth}), the following soft symbols are also used when needed. For example, $\tilde O(n)$ (soft-$O$) means $O(n\,{\rm polylog}\, n)$, i.e., $O(n\log^k n)$ for some $k$, while $\tilde o(n)$ (soft-$o$) means $o(n\,{\rm polylog}\, n)$, i.e., $o(n\log^k n)$ for any $k$. Soft-$\Omega$ and soft-$\omega$ are defined analogously.

Now, we introduce the notion of infinite asymptotic gaps and discuss different types of such gaps in an intuitive manner. Formal definitions are left to Appendix \ref{app:gap}.
The gap between two asymptotic complexities is normally characterized by a type of increasing monotone function. For example, there is a quadratic gap between $O(\sqrt{n})$ and $\Omega(n)$, and an exponential gap between $O(\log n)$ and $\Omega(n)$. However, when one side is $\omega(1)$, i.e., superconstant (or $o(1)$, i.e., subconstant), while the other side is not, the gap between them grows faster than any such monotone function. We regard such gaps as infinite.
In fact, all (positive) asymptotic complexities belong to one of the following three classes: $o(1)$, $\Theta(1)$, or $\omega(1)$. In the logarithmic scale, these three types of asymptotic complexities tend to negative infinity, constant, and positive infinity, respectively. The gap between any two of them is infinite. In particular, an $o(1)$ vs.~$\omega(1)$ gap can be regarded as doubly infinite, whereas an $o(1)$ vs.~$\Theta(1)$ or $\Theta(1)$ vs.~$\omega(1)$ gap is only singly infinite. Evidently, the gap between any two asymptotic complexities cannot be larger than doubly infinite. The general behaviors and comparisons of infinite gaps are illustrated in Fig.~\ref{infgap}.


\begin{figure}
\centering
\includegraphics[width=\columnwidth]{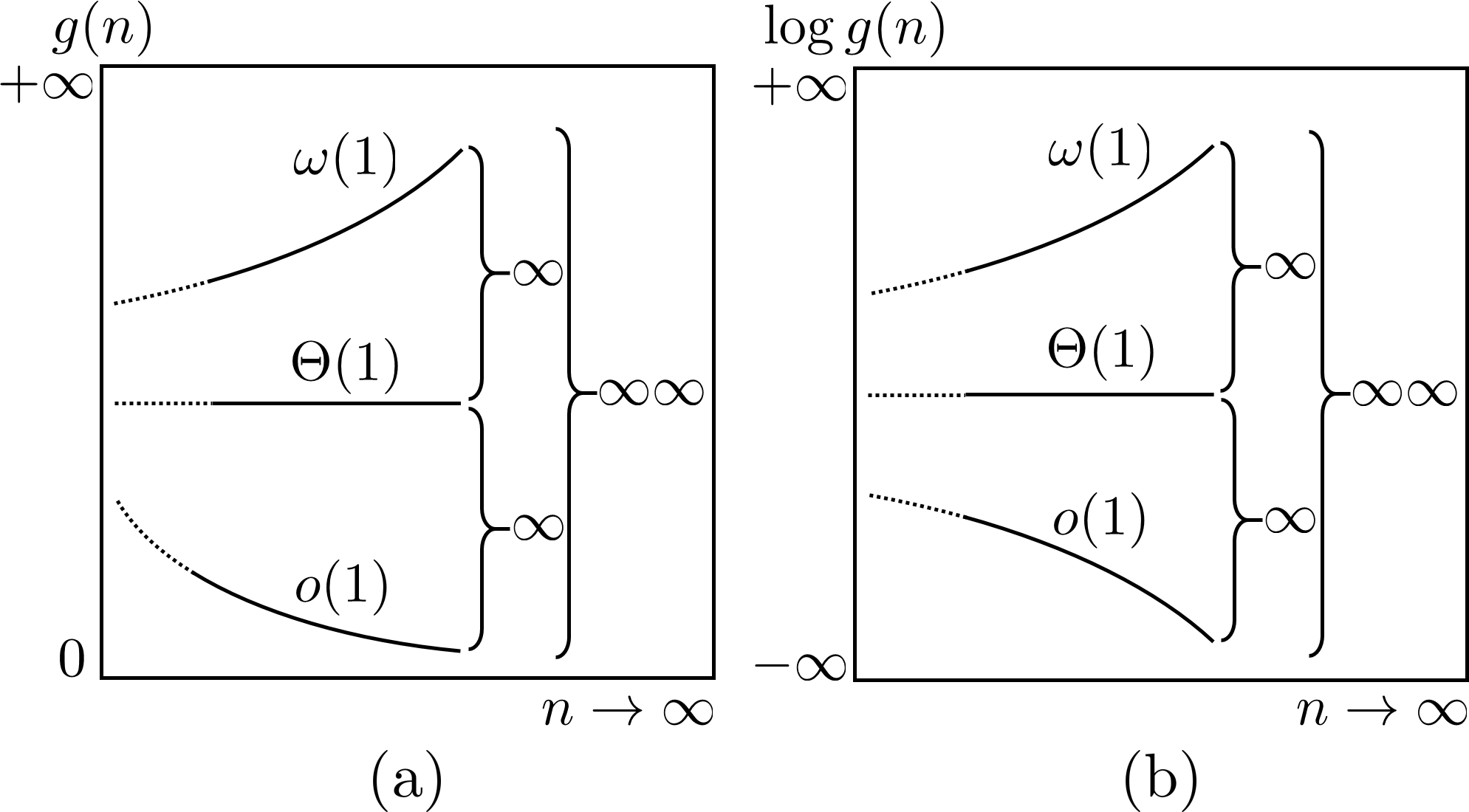}
\caption{Illustrations of infinite gaps in the (a) linear scale and (b) logarithmic scale. ``$\infty$'' means singly infinite; ``$\infty\infty$'' means doubly infinite.}
\label{infgap}
\end{figure}



\section{Exclusion game}


A communication task between two players is typically defined by a function $f :\{0,1\}^* \times \{0,1\}^* \rightarrow \{0,1\}^*$ ($\{0,1\}^*$ denotes the set of bit strings with arbitrary length): Alice and Bob are respectively given some string $ x\in \{0,1\}^*$ and $ y \in \{0,1\}^*$. They are allowed to exchange messages, and one of them outputs a string $z\in \{0,1\}^*$ in the end. They succeed at the task if $z = f(x,y)$. 

To formally define the exclusion game, a more general framework of communication tasks is needed. The problem is now defined by a relation $R\subseteq\{0,1\}^* \times \{0,1\}^* \times \{0,1\}^*$. Furthermore, we restrict the protocol to be one-way: Alice can send one message to Bob, and Bob outputs an answer. They succeed at the task if $(x,y,z) \in R$.
It is evident that the general framework reduces to the original one if for all $x,y \in \{0,1\}^*$, there exists a unique $z\in \{0,1\}^*$ for which $(x,y,z) \in R$. Generically, relational tasks are those admitting multiple winning outputs for one certain input.

The exclusion game is a relational task defined by the relation $R_{\EXC} = \{(x,y,z)|z \neq \mathcal M_y( x)\}$: Alice's input $x\in\{0,1\}^n$ and Bob's input $y \subseteq [n], |y|=m$ ($y$ can be encoded as a string to conform to the above general framework) are both drawn randomly from uniform distributions, and $\mathcal M_y( x)$ denotes the string given by $x$ restricted to the bits specified by $y$. The winning condition is that Bob's output $ z \neq \mathcal M_y( x)$, for given $x$ and $y$.




PJO devised the following quantum strategy \cite{comm} that wins every exclusion game with certainty, i.e., works for any $\EXC_{n,m,\gamma}$.
Given the input $ x=x_1 \cdots x_n$, Alice encodes each classical bit $x_i$ as the qubit
\begin{equation}
|\phi(x_i;\theta_m)\rangle = \cos \left(\frac {\theta_m} 2\right) |0\rangle + (-1)^{x_i} \sin\left(\frac{\theta_m} 2\right) |1\rangle ,
\end{equation}
where $\theta_m = 2 \tan^{-1} (2^{1/m} - 1)$. The $n$-bit string $ x$ is then encoded as the joint state
\begin{equation}
|{\Phi(x;\theta_m)}\rangle = \bigotimes_{i=1}^n |{\psi(x_i;\theta_m)}\rangle,\label{PBRstate}
\end{equation}
which she sends to Bob via the quantum channel.
Upon receiving the state from Alice, Bob performs a global measurement across the $m$ systems specified by $y$ (denoted by $|\Psi(\mathcal M_y( x);\theta_m)\rangle$). The measurement is given by
\begin{equation}
|{\zeta(z)}\rangle = \frac 1{\sqrt{2^m}} \left(|0\rangle - \sum_{ s \neq  0 } (-1)^{ z\cdot  s} | s\rangle\right).
\end{equation}
As one can verify, $\langle\Phi(\mathcal M_y( x);\theta_m)|{\zeta(\mathcal M_y( x))}\rangle=0$ \cite{exclusion}. That is, Bob always outputs some $ z\neq{\mathcal M_y( x)}$ according to the measurement outcome. Therefore, they win the game with certainty.
 This measurement technique may be described as a conclusive-exclusion measurement. It was first introduced in Ref.\ \cite{Caves}, and was subsequently used to prove the Pusey-Barrett-Rudolph (PBR) theorem \cite{pbr}, a result in the field of quantum foundations that rules out a certain class of $\psi$-epistemic models of quantum mechanics.

The PJO strategy exhibits a striking property: The amount of information Alice actually reveals to Bob (the information cost) tends to zero as $n$ increases, in a certain regime. More specifically, it can be calculated that $Q_{IC} ({\rm PJO}) \leq 2 S(M_Q) \in O(nm^{-2}\log m)$, where $S(M_Q)$ is the von Neumann entropy of the quantum message $M_Q$ (the ensemble of $|{\Phi(x;\theta_m)}\rangle$ where $x$ is an $n$-bit string with each of the $2^n$ possibilities being equally likely) that Alice sends to Bob. When $m\in\omega(\sqrt{n\log n})$, $\lim_{n\rightarrow\infty} Q_{IC}({\rm PJO})=0$. This directly indicates that $\lim_{n\rightarrow\infty}\bar Q_{IC}(\EXC_{n,m,0})=0$ in the specified regime. Note that this actually holds for any prior distribution of inputs \cite{comm}.


\section{Zero-error quantum communication complexity}
Here, we prove an $\Omega(\log n)$ lower bound on $\bar Q_{CC}(\EXC_{n,m,0})$, when $m\in \tilde  o(n)$. That is, there cannot exist any winning quantum strategy whose communication cost scales sublogarithmically in $n$ in this regime.

The main idea of the proof is to approximately simulate any quantum protocol for $\EXC_{n,m,0}$ by a classical protocol with exponential overhead, and show that the task can still be accomplished with zero probability of error. Because of the tiny possible probability of error associated with the exclusion game, the existence of such a simulation is non-obvious and the results of Ref.\ \cite{kremer} cannot be applied, but we show that it can be made to work when $ m\in \tilde o(n) $. Then, lower bounds on the classical communication complexities in this regime would directly imply exponentially smaller lower bounds on the corresponding quantum communication complexities.
The following lemma sets a linear lower bound on the classical communication complexities of almost all exclusion games:
\begin{lem} \label{classical_lower_bound}
For $m\leq\alpha n$ where $0<\alpha<1/2$ is a constant, $\bar C_{CC}(\EXC_{n,m,0}) \in \Omega(n)$. 
\end{lem}
The proof of this lemma is given in Appendix \ref{app:lemmas}. Note that the applicable regime of this lemma covers $m\in \tilde  o(n)$. This enables us to prove the following result:
\begin{thm}\label{main}
For $m\in \tilde  o(n)$, $\bar Q_{CC}(\EXC_{n,m,0}) \in \Omega(\log n)$. 
\end{thm}
\begin{proof}
Here, we only sketch the main steps of the proof. Details are given in Appendix \ref{app:zero}. 
Suppose that for $\EXC_{n,m,0}$ with $m\in \tilde  o(n)$, there exists a winning quantum strategy $\Pi_Q$ such that $Q_{CC}(\Pi_Q)\equiv q\in o(\log n)$. Then based on $\Pi_Q$, we can devise a corresponding classical strategy $\Pi_C$ such that $C_{CC}(\Pi_C)\in o(n)$ as follows. Given input $x$, the quantum message that Alice sends to Bob in $\Pi_Q$ can be encoded as a $2q$-qubit pure state $|\psi(x)\rangle$.
First, Alice prepares a classical message $C(|\psi(x)\rangle)$ that approximately encodes $|\psi(x)\rangle=\sum^{2^{2q}}_{j=1} \alpha_j|j\rangle=\sum^{2^{2q}}_{j=1}(b_j+ic_j)|j\rangle$, by registering the real ($b_j$) and imaginary parts ($c_j$) of all amplitudes ($\alpha_j$) to accuracy $2^{-(m+q)}/20$ (the approximations are denoted by $\tilde b_j$ and $\tilde c_j$). It can be shown that $|C(|\psi(x)\rangle)|\in o(n)$, when $m\in \tilde  o(n)$. Alice then sends $C(|\psi(x)\rangle)$ to Bob, whose local strategy can be considered as a positive-operator valued measure (POVM) $\{P_{ z}\}$ with $2^m$ elements, each indicating an $m$-bit output string $ z$.
Bob first normalizes the amplitude vector encoded in $C(|\psi(x)\rangle)$, and then applies Born's rule to compute the approximate probability $p_{ z}$ of obtaining each $ z$. Given the above accuracy of encoding, it can be shown that $p_{\mathcal{M}_y( x)}<2^{-m}$. Therefore, Bob simply outputs a $ z$ such that $p_{ z}>2^{-m}$, which always exists. Since $C_{CC}(\Pi_C)\in o(n)$, we have reached a contradiction to Lemma \ref{classical_lower_bound}. Therefore, no quantum strategies $\Pi_Q$ such that $Q_{CC}(\Pi_Q)\in o(\log n)$ for $\EXC_{n,m,0}$ with $m\in \tilde  o(n)$ can exist: $\bar Q_{CC}(\EXC_{n,m,0}) \in \Omega(\log n)$ in this regime. 
\end{proof}

This result directly implies the following gaps between complexities:
\begin{cor}\label{gaps}
For $\EXC_{n,m,0}$ with $m\in\omega(\sqrt{n\log n})$ and $m\in \tilde  o(n)$, we have $\bar Q_{IC}\in O(nm^{-2}\log m)$ (tends to zero as $n$ increases), $\bar Q_{CC}\in\Omega(\log n)$ and $\bar C_{CC}\in\Theta(n)$. Thus the gap between
\begin{itemize}
\item $\bar Q_{IC}$ and $\bar Q_{CC}$: doubly infinite;
\item $\bar Q_{CC}$ and $\bar C_{CC}$: at most exponential.
\end{itemize}
\end{cor}


\section{Robustness against error}
In the discussions above, Bob is required to output a right answer every single time. If error is allowed sometimes, do the properties of the zero-error instances still hold? Note that $\gamma\geq 2^{-m}$ is trivial since such probability of error can be achieved by randomly guessing without any communication. With a variant of the zero-error simulation protocol, we show the following general result for $\gamma<2^{-m}$:
\begin{thm}\label{error}
Consider some $h(m)$ such that $\gamma$ satisfies $-\log(2^{-m}-\gamma)\in O(h(m))$. Suppose that for $\EXC_{n,m,\gamma}$ with $\gamma<2^{-m}$, there exists a winning quantum strategy $\Pi^\gamma_Q$ such that $Q_{CC}(\Pi^\gamma_Q)\equiv s\in O(\xi(n))$. Then, one can construct a classical strategy $\Pi^{0^+}_C$ such that $C_{CC}(\Pi^{0^+}_C)=[O(h(m))+O(\xi(n))]2^{O(\xi(n))}$, whose probability of error can be made arbitrarily small. 
\end{thm}
\begin{proof}
Here, we only sketch the main steps of the proof. Details are given in Appendix \ref{appb}.  
We revise Bob's local part of $\Pi_C$ presented in Theorem \ref{main} to devise $\Pi^{0^+}_C$ as follows. As for the zero-error case, given input $x$, Alice prepares an $[O(h(m))+O(\xi(n))]2^{O(\xi(n))}$-bit classical message that encodes the real and imaginary parts of all amplitudes of the quantum message $|\psi^\gamma(x)\rangle$ in $\Pi^\gamma_Q$ to accuracy $(2^{-m}-\gamma)2^{-s}/20$, and sends it to Bob, who then normalizes the amplitude vector. 
Instead of classically calculating the probability distribution of the output as in $\Pi_C$, Bob now resorts to local quantum resources. He simply prepares a new quantum state $|\tilde\psi^\gamma(x)\rangle$ according to the amplitudes, and then feeds it into his original local quantum computation. It can be shown that the probability of outputting the wrong answer $p_{\mathcal{M}_y({ x})}$ is always less than $2^{-m}$, which guarantees that $\mathcal{M}_y({ x})$ is not the winning output.
Therefore, Bob can run $\Pi^{0^+}_C$ multiple times and take majority vote to suppress the probability of error to an arbitrarily small value by the Chernoff bound (amplitude amplification). 
\end{proof}
When $\gamma=2^{-(m+1)}, m\geq\sqrt{n}$, it was shown in an early version of Ref.\ \cite{comm} that only one classical bit of communication is needed. For completeness we include the proof in Appendix \ref{appb}. Therefore we consider only the regime of even smaller $\gamma$ to be of interest. Since $m<-\log(2^{-m}-\gamma)< m+1$ under this constraint, $h(m)$ can be replaced by $m$ in the above discussions. Like the zero-error case, for $m\in \tilde  o(n)$, this theorem indicates that the gap between $\bar Q_{CC}$ and $\bar C_{CC}$ for $\EXC_{n,m,\gamma}$, when any non-trivial $\gamma$ is allowed, is at most exponential. Consequently, the logarithmic lower bound on $\bar Q_{CC}$ and the gaps established for the zero-error case still hold even if some $\gamma$ such that $\bar C_{CC}\in\Omega(n)$ is allowed. The permissible range of $\gamma$ is identified by the following theorem:
\begin{thm}\label{Cerror}
For $m\leq\alpha n$ where $0<\alpha<1/2$ is a constant and $ \gamma\leq (n+1)^{-m} $, $\bar C_{CC}(\EXC_{n,m,\gamma})\in\Omega(n) $.
\end{thm}
The proof of Theorem \ref{Cerror} is given in Appendix \ref{app:Cerror}. Combining Theorems \ref{error} and \ref{Cerror}, we obtain the following results:
\begin{cor}\label{gaperror}
For $m\in \tilde  o(n)$ and $ \gamma\leq (n+1)^{-m} $, $\bar Q_{CC}(\EXC_{n,m,\gamma})\in\Omega(\log n) $. By further restricting $m\in\omega(\sqrt{n\log n})$, the gaps established in Corollary \ref{gaps} still hold.
\end{cor}


\section{Compressing quantum communication}
Although the PJO strategy succeeds with vanishingly small amount of information cost, it requires exactly $n$ qubits of communication, which is maximal. 
For $m\in \tilde  o(n)$, the possibility of superexponential compression of quantum communication cost has been ruled out, but it remains unsettled if any compression is possible at all. In particular, one may wonder if quantum strategies can be more efficient than classical ones in communication cost. 
Here, we show that a polynomial reduction of quantum communication cost can be achieved by abandoning an insignificant part of the PJO message, while only causing a tiny probability of error such that $\bar C_{CC}\in\Omega(n)$ still holds:
\begin{thm}\label{compression}
For $ m\in \Theta(n^\alpha) $, $ 1/2<\alpha<1 $ and $ \gamma\geq(n+1)^{-m}$, $\bar Q_{CC}(\EXC_{n,m,\gamma})\in O(m^{1+\delta}) $, $ \delta>0 $.
\end{thm}
\begin{proof}
Here, we only sketch the main steps of the proof. Details are given in Appendix \ref{appd}. 
Instead of directly sending the $n$-qubit state given by Eq. \eqref{PBRstate}, Alice now compresses the message by projecting it onto the subspace spanned by the computational basis vectors with Hamming weight (the number of ones) at most $k$.
Upon receiving the message, Bob performs the same measurement on the quantum state as in the original PJO strategy. Obviously, this would lead to some probability of error $\epsilon_k$. However, it can be shown that taking $ k=m^{1+\eta} $ with any $ \eta>0 $ is sufficient to guarantee that $\epsilon_k\leq(n+1)^{-m}$.
It then follows that the size of the compressed message scales as $ O(m^{1+\delta}) $ for any $ \delta>0 $.
\end{proof}
Combining Theorems \ref{Cerror} and \ref{compression}, we obtain another quantum-classical separation:
\begin{cor}
For $\EXC_{n,m,\gamma}$ with $\gamma\geq\left(n+1\right)^{-m}$, $m\in\Theta\left(n^{\alpha}\right)$, $1/2<\alpha<1$, there is a polynomial gap between $\bar Q_{CC}$ and $\bar C_{CC}$.
\end{cor}


\section{Concluding remarks}
In this paper, we obtained some new knowledge about quantum communication by studying different regimes of the exclusion game.
The key result of this paper is a logarithmic lower bound on the quantum communication complexity of most exclusion games. This bound indicates the following results: (i) a doubly infinite gap between the quantum information and communication complexities; (ii) the gap between the quantum and classical communication complexities of certain relational tasks with exponentially small possible probability of error is at most exponential. In contrast, the largest known gap between classical information and communication complexities is singly infinite \cite{Braverman}, and the known upper bound of quantum-classical gap in communication complexities only applies to bounded-error function problems \cite{kremer}.
For exclusion games, we leave open the problems of whether the $\Omega(\log n)$ lower bound on the quantum communication complexity for $m\in \tilde o(n)$ is tight, and whether the gap between the quantum and classical communication complexities for $m\in \tilde \Omega(n)$ can be superexponential. (Interestingly, for a slight modification of the exclusion game, there exists a singly infinite gap between the entanglement-assisted communication complexity and the ordinary classical communication complexity \cite{comm}.) Another set of important open problems is related to how large the gap between quantum information and communication complexities can be in different settings, e.g., bounded-error, entanglement-assisted, interactive. Answers to these problems will provide significant insight into the relation between these complexities, and the power of quantum resources in the communication model.
 




\begin{acknowledgements}
The authors would like to thank R.\ Jain and J.\ Oppenheim for insightful discussions, and the anonymous referee for helpful comments. Z.W.L.\ is supported by an Army Research Office award (No.\ W911NF-11-1-0400). C.P. is supported by the European Research Council (ERC Grant Agreement No.\ 337603). Y.Z.\ is supported by the ARO grant Contract No.\ W911NF-12-0486. D.E.K.\ is supported by the National Science Scholarship from the Agency for Science, Technology and Research (A*STAR). S.A. is supported by an Alan T. Waterman Award from the National Science Foundation, under Grant No.\ 1249349.
\end{acknowledgements}

%

\widetext
\appendix

\makeatletter

\section{Formal definitions of infinite gaps}\label{app:gap}
Here, we formally define and classify infinite gaps between two positive asymptotics, $g_1(n)$ and $g_2(n)$ (without loss of generality, assume that $\lim\limits_{n\to\infty}g_2(n)/g_1(n)\geq 1$), as $n$ tends to infinity. 
The key idea of properly characterizing all possible gaps is to symmetrize the increasing and decreasing asymptotics by using the logarithmic scale. As discussed in the main text, finite gaps are characterized by a type of well-behaved increasing monotone function.
However, there exist gaps that are larger than any finite one, in the sense that they grow faster than any monotone function asymptotically:
\begin{defn}[Infinite gap]
The gap between $g_1(n)$ and $g_2(n)$ is infinite, if there does not exist any strictly increasing monotone function $g$ such that $\lim\limits_{n\to\infty}\log g_2(n)/g(\log g_1(n))=1$.
\end{defn}
\noindent 
Infinite gaps can be further classified:
\begin{defn}[Doubly infinite gap]
The gap between $g_1(n)$ and $g_2(n)$ is doubly infinite, if there exists an intermediate asymptotic $g_m(n)$ satisfying $\lim\limits_{n\to\infty}g_2(n)/g_m(n)\geq 1,\lim\limits_{n\to\infty}g_m(n)/g_1(n)\geq 1$ such that $g_1(n)$ vs. $g_m(n)$ and $g_m(n)$ vs. $g_2(n)$ are both infinite gaps.
\end{defn}
\begin{defn}[Singly infinite gap]
The gap between $g_1(n)$ and $g_2(n)$ is singly infinite, if $g_1(n)$ vs. $g_2(n)$ is an infinite gap, but there does not exist any asymptotic $g_m(n)$ satisfying $\lim\limits_{n\to\infty}g_2(n)/g_m(n)\geq 1,\lim\limits_{n\to\infty}g_m(n)/g_1(n)\geq 1$ such that $g_1(n)$ vs. $g_m(n)$ and $g_m(n)$ vs. $g_2(n)$ are both infinite gaps.
\end{defn}
\noindent It is evident that, if the gap between $g_1(n)$ and $g_2(n)$ is doubly infinite, the only possibility is that $g_m(n)\in\Theta(1),g_1(n)\in o(1),g_2(n)\in\omega(n)$. This is the largest type of gap between two positive asymptotics. 
Gaps that take the form $o(n)$ vs. $\Theta(1)$ or $\Theta(1)$ vs. $\omega(n)$ are singly infinite. Infinite gaps are either singly infinite or doubly infinite.

\section{Lemmas for Theorems \ref{main} and \ref{error}}\label{app:lemmas}
Here, we present the detailed proofs of some lemmas that are useful for proving Theorems \ref{main} and \ref{error}, including Lemma \ref{classical_lower_bound}, which has already been stated in the main text. 

\theoremstyle{plain}
\newtheorem*{lem:classical_lower_bound}{Lemma \ref{classical_lower_bound}}
\begin{lem:classical_lower_bound}
For $m\leq\alpha n$ where $0<\alpha<1/2$ is a constant, $\bar C_{CC}(\EXC_{n,m,0}) \in \Omega(n)$.
\end{lem:classical_lower_bound}
\begin{proof}
By Theorem 2 of Ref.\ \cite{comm}, for any classical strategy $\Pi$ that wins $\EXC_{n,m,0}$, $C_{IC}(\Pi)\geq n-\log_2\left(\sum_{i=0}^{m-1}{n \choose i}\right)$. For $m\leq\alpha n$ where $0<\alpha<1/2$ is a constant, $C_{IC}(\Pi)\in\Omega(n)$ (see Appendix C of Ref.\ \cite{comm}).
Since the amount of information revealed cannot exceed the amount of communication, i.e., $C_{IC}\leq C_{CC}$ for any communication protocol \cite{cywy,ccic}, it follows that $C_{CC}(\Pi)\in \Omega(n)$. Note that Alice can always send the whole string to Bob in order to win, thus in fact $C_{CC}(\Pi)\in \Theta(n)$. Therefore, $\bar C_{CC}(\EXC_{n,m,0}) \in \Omega(n)$ for the specified regime of $m$ asymptotically.
\end{proof}

\begin{lem}\label{cbits}
A $t$-qubit pure quantum state can be classically described by a set of real numbers encoding the real and imaginary parts of all amplitudes to accuracy $\epsilon$ using $O\left(2^t\log({1}/{\epsilon})\right)$ bits.
\end{lem}
\begin{proof}
Generically, a $t$-qubit pure state $|\psi_t\rangle$ can be written as $|\psi_t\rangle=\sum_{i=1}^{2^t}\alpha_i|i\rangle$, where $\alpha_i\in\mathbb{C}$, and $\{|i\rangle\}$ is a complete orthonormal basis set containing $2^t$ elements. We express all complex amplitudes as $\alpha_i=b_i+ic_i$ where $b_i, c_i\in\mathbb{R}$, satisfying $\sum_{i=1}^{2^t}|\alpha_i|^2=\sum_{i=1}^{2^t}(b_i^2+c_i^2)=1$. Thus, $0\leq|b_i|,|c_i|\leq 1$. To approximate each of these real numbers to accuracy $\epsilon=2^{-r}$, we keep the first $r$ bits after the binary point, and use one extra bit to indicate its sign, i.e., we can find an $(r+1)$-bit classical string that encodes an approximation $\tilde b_i$ of each $b_i$ such that for all $i$,
\begin{equation}
\begin{split}
\Delta b_i=|\tilde b_i-b_i|\leq\epsilon,\\
\Delta c_i=|\tilde c_i-c_i|\leq\epsilon.
\end{split}
\end{equation}
Notice that there are $2\times 2^t$ such numbers in total, thus only $2^{t+1}(r+1)=O\left(2^t\log({1}/{\epsilon})\right)$ bits are needed to encode $|\psi_t\rangle$ such that we have specified the real and imaginary parts of all amplitudes to accuracy $\epsilon$. 
\end{proof}

\begin{lem}\label{pert}
Let $ \mathcal{H} $ be a Hilbert space of dimension $ |\mathcal{H}|=l $, with orthonormal basis $\{ |1 \rangle,\cdots,| l \rangle\}$. Let $ |\psi \rangle \in \mathcal{H} $ with $ |\psi \rangle=\sum_{j=1}^{l}\alpha_j|j \rangle=\sum_{j=1}^{l}(b_j+ic_j)|j \rangle $. Suppose that we have $ \{\tilde{b}_j,\tilde{c}_j\} $ such that $\forall j$, $ |b_j-\tilde{b}_j|,|c_j-\tilde{c}_j|\leq\epsilon <(6\sqrt{2l})^{-1}$. Let $ |\tilde{\psi}\rangle = \sum_{j=1}^{l}\tilde{\alpha}_j|j \rangle=\sum_{j=1}^{l}(\tilde{b}_j+i\tilde{c}_j)|j \rangle/\nu$ where $\nu\equiv\sqrt{{\sum_{k=1}^l(\tilde{b}_k^2+\tilde{c}_k^2)}} $ is the norm. Then $ D(|\psi\rangle,|\tilde{\psi}\rangle)< 10\sqrt{l}\epsilon $, where $ D(\,,\,) $ is the trace distance.
\end{lem}
\begin{proof}
We first consider the normalization factor:
\begin{equation}
\nu^2\equiv\sum^l_{j=1}\tilde{b}_j^2+\tilde{c}_j^2\leq\sum^l_{j=1}(|b_j|+\epsilon)^2+(|c_j|+\epsilon)^2=1+2\sum^l_{j=1}(|b_j|+|c_j|)\epsilon+2l\epsilon^2.
\end{equation}
By the Cauchy-Schwarz inequality, we have 
\begin{equation}
\sum^l_{j=1}|b_j|+|c_j|\leq\sqrt{2l},\label{Cauchy}
\end{equation}
and 
\begin{equation}
2l\epsilon^2<2l\frac{1}{\sqrt{2l}}\epsilon=\sqrt{2l}\epsilon.
\end{equation}
Therefore,
\begin{equation}
\nu^2<1+3\sqrt{2l}\epsilon.
\end{equation}
Similarly,
\begin{equation}
1-2\sqrt{2l}\epsilon<\nu^2.
\end{equation}
Since 
\begin{equation}
\frac{1}{\sqrt{1+3\sqrt{2l}\epsilon}}>\sqrt{1-3\sqrt{2l}\epsilon}>1-3\sqrt{2l}\epsilon,
\end{equation}
and
\begin{equation}
\frac{1}{\sqrt{1-2\sqrt{2l}\epsilon}}<\sqrt{1+3\sqrt{2l}\epsilon}<1+3\sqrt{2l}\epsilon,
\end{equation}
we have
\begin{equation}
1-3\sqrt{2l}\epsilon<\frac{1}{\nu}<1+3\sqrt{2l}\epsilon.
\end{equation}
Assuming $ b_j>0 $,
then 
\begin{eqnarray}
(b_j-\epsilon)(1-3\sqrt{2l}\epsilon)<\frac{\tilde{b}_j}{\nu}<(b_j+\epsilon)(1+3\sqrt{2l}\epsilon)~~ & \text{if}~~b_j-\epsilon>0, \\
(b_j-\epsilon)(1+3\sqrt{2l}\epsilon)<\frac{\tilde{b}_j}{\nu}<(b_j+\epsilon)(1+3\sqrt{2l}\epsilon)~~ & \text{if}~~b_j-\epsilon<0.
\end{eqnarray}
For both cases,
\begin{equation}
(b_j+\epsilon)(1+3\sqrt{2l}\epsilon)= b_j+(1+3\sqrt{2l}b_j)\epsilon+3\sqrt{2l}\epsilon^2<b_j+(2+3\sqrt{2l}b_j)\epsilon.
\end{equation}
For $b_j-\epsilon>0$,
\begin{equation}
(b_j-\epsilon)(1-3\sqrt{2l}\epsilon)>b_j-(1+3\sqrt{2l}b_j)\epsilon.
\end{equation}\\
For $b_j-\epsilon<0$,
\begin{equation}
(b_j-\epsilon)(1+3\sqrt{2l}\epsilon)=b_j-(1-3\sqrt{2l}b_j)\epsilon-3\sqrt{2l}\epsilon^2>b_j-(2+3\sqrt{2l}b_j)\epsilon.
\end{equation}
So if $ b_j>0 $,
\begin{equation}
\left|b_j-\frac{\tilde{b}_j}{\nu}\right|<(2+3\sqrt{2l}b_j)\epsilon.
\end{equation}
Similarly, if $ b_j<0 $,
\begin{equation}
\left|b_j-\frac{\tilde{b}_j}{\nu}\right|<(2-3\sqrt{2l}b_j)\epsilon.
\end{equation}\\
So we obtain
\begin{equation}
\left|b_j-\frac{\tilde{b}_j}{\nu}\right|<(2+3\sqrt{2l}|b_j|)\epsilon.
\end{equation}
Similarly,
\begin{equation}
\left|c_j-\frac{\tilde{c}_j}{\nu}\right|<(2+3\sqrt{2l}|c_j|)\epsilon.
\end{equation}
Recall that $|\tilde{\psi}\rangle=\sum\tilde{\alpha}_j|j\rangle$, where $\tilde{\alpha_j}=(\tilde{b}_j+i\tilde{c}_j)/\nu$. Then
\begin{eqnarray}
\left|\alpha_j-\tilde{\alpha}_j\right|&=&\left|b_j+ic_j-\frac{\tilde{b}_j}{\nu}-i\frac{\tilde{c}_j}{\nu}\right|\nonumber\\
&\leq&\left|b_j-\frac{\tilde{b}_j}{\nu}\right|+\left|c_j-\frac{\tilde{c}_j}{\nu}\right|\nonumber\\
&<&(4+3\sqrt{2l}(|b_j|+|c_j|))\epsilon.
\end{eqnarray}
Therefore,
\begin{eqnarray}
1-\left|\langle\psi|\tilde{\psi}\rangle\right|^2
&=&\left(1-\left|\langle\psi|\tilde{\psi}\rangle\right|\right)\left(1+\left|\langle\psi|\tilde{\psi}\rangle\right|\right)\nonumber\\
&\leq&2\left(1-\left|\langle\psi|\tilde{\psi}\rangle\right|\right)\nonumber\\
&\leq&2-\langle\psi|\tilde{\psi}\rangle-\langle\tilde{\psi}|\psi\rangle\nonumber\\
&=&\sum^l_{j=1}|\alpha_j|^2+|\tilde{\alpha}_j|^2-\alpha_j\tilde{\alpha}_j^*-\alpha_j^*\tilde{\alpha}_j\nonumber\\
&=&\sum^l_{j=1}|\alpha_j-\tilde{\alpha}_j|^2\nonumber\\
&<&\sum^l_{j=1}(4+3\sqrt{2l}(|b_j|+|c_j|))^2\epsilon^2,
\end{eqnarray}
where $(4+3\sqrt{2l}(|b_j|+|c_j|))^2=16+24\sqrt{2l}(|b_j|+|c_j|)+18l(|b_j|+|c_j|)^2$.
Using $(|b_j|+|c_j|)^2\leq2(|b_j|^2+|c_j|^2)$ and Eq. \eqref{Cauchy}, we obtain
\begin{eqnarray}
1-\left|\langle\psi|\tilde{\psi}\rangle\right|^2
&<&\sum^l_{j=1}(16+24\sqrt{2l}(|b_j|+|c_j|)+36l(|b_j|^2+|c_j|^2))\epsilon^2\nonumber\\
&\leq&(16l+48l+36l)\epsilon^2=100l\epsilon^2.
\end{eqnarray}
Then,
\begin{equation}
D(|\psi\rangle,|\tilde{\psi}\rangle)=\sqrt{1-\left|\langle\psi|\tilde{\psi}\rangle\right|^2}<10\sqrt{l}\epsilon.
\end{equation}\\
Thus, $D(|\psi\rangle,|\tilde{\psi}\rangle)<10\sqrt{l}\epsilon$.
\end{proof}

\begin{lem}\label{pert2}
Let $ \{P_k\} $ be a POVM, with $ p_k=\langle\psi| P_k|\psi\rangle $, $ \tilde{p}_k=\langle\tilde{\psi}| P_k|\tilde{\psi}\rangle $. Then 
$ |p_k-\tilde{p}_k|<20\sqrt{l}\epsilon $.
\end{lem}
\begin{proof}
By Theorem 9.1 in \cite{nc}, we directly obtain
 $ 
 |p_k-\tilde{p}_k|\leq \sum^l_{k=1}|p_k-\tilde{p}_k|\leq 2D(|\psi\rangle,|\tilde{\psi}\rangle)<20\sqrt{l}\epsilon $, where the last step comes from Lemma \ref{pert}.
\end{proof}

\section{Detailed Proof of Theorem \ref{main}}\label{app:zero}

\newtheorem*{thm:main}{Theorem \ref{main}}
\begin{thm:main}
For $m\in \tilde  o(n)$, $\bar Q_{CC}(\EXC_{n,m,0}) \in \Omega(\log n)$. 
\end{thm:main}
\begin{proof}
Suppose that for $\EXC_{n,m,0}$ where $m\in \tilde  o(n)$, there exists a winning quantum strategy $\Pi_Q$ such that $Q_{CC}(\Pi_Q)\equiv q\in o(\log n)$. $q=\log|\mathcal{H}|$, where $\mathcal{H}$ is the Hilbert space of the largest quantum message. Then based on $\Pi_Q$, we can devise a corresponding classical strategy $\Pi_C$ with $o(n)$ bits of communication, which contradicts Lemma \ref{classical_lower_bound}, therefore negating the existence of $\Pi_Q$.

Most generally, $\Pi_Q$ can be divided into three steps: (i) Alice prepares a quantum message (state) of size at most $q$, based on her $n$-bit string $ x$; (ii) Alice sends the state to Bob; (iii) Bob feeds the state into his local quantum computation, and obtains an $m$-bit string $ z$ such that $ z\neq \mathcal M_y( x)$ according to the output (measurement outcome). 
Note that the quantum messages can in general be mixed, but each of them can always be encoded as a $2q$-qubit pure state by purification using an ancilla space of $q$ qubits (append maximally mixed qubits whn the original state contains less than $q$ qubits). Denote the pure message corresponding to $x$ as $|\psi(x)\rangle$.
In addition, both players agree on a fixed basis for the matrix representation of operators and amplitudes of state vectors beforehand.
 
 The essence of constructing $\Pi_C$ is to classically simulate all steps of $\Pi_Q$. The basic procedure goes as follows.
First, Alice prepares a classical message $C(|\psi_x\rangle)$ that approximately encodes $|\psi_x\rangle=\sum^{2^{2q}}_{j=1} \alpha_j|j\rangle=\sum^{2^{2q}}_{j=1}(b_j+ic_j)|j\rangle$ ($\{|j\rangle\}$ is the predetermined basis), by registering the real ($b_j$) and imaginary parts ($c_j$) of all amplitudes ($\alpha_j$) to some desired accuracy $\bar\epsilon$ (the approximations are denoted by $\tilde b_j$ and $\tilde c_j$), and then sends it to Bob. Note that the size of $C(|\psi_x\rangle)$, i.e., the communication cost of $\Pi_C$, depends on $\bar\epsilon$: it grows as higher precision is desired. In $\Pi_Q$, Bob's local strategy can always be modeled as a quantum circuit with $|\psi(x)\rangle$ being the input, i.e., quantum operations followed by a generalized measurement by the principle of deferred measurement \cite{nc}, which is altogether equivalent to some POVM $\{P_i\}$.
Although $\{P_i\}$ may contain an arbitrary number of elements in principle, there are only $2^m$ possible strings that Bob can eventually output: $g(P_i)=z$, where $z$ is an $m$-bit string. Therefore, all $P_i$'s corresponding to the same $z$ can be combined as an element $P'_z$ of a new POVM $\{P'_z\}$ by
\begin{equation}
P'_z = \sum_{g(P_i)=z}P_i,
\end{equation}
or in the continuum limit where the elements are labeled by a continuous variable $\mu$,
\begin{equation}
P'_z = \int_{g(P(\mu))=z}d\mu P(\mu).
\end{equation}
Due to the convexity of the set of all non-negative Hermitian operators (valid POVM elements), $\{P'_z\}$ forms a discrete effective POVM with $2^m$ elements labeled by $z$.
A subtlety here is that the amplitude vector encoded in $C(|\psi_x\rangle)$ is not necessarily normalized. Bob first normalizes the amplitude vector by dividing each component with the 2-norm $\nu\equiv{\sqrt{\sum_{j=1}^{2^{2q}}(\tilde{b}_j^2+\tilde{c}_j^2)}} $, and then applies Born's rule to compute the approximate probability of obtaining each $z$:
\begin{equation}\label{ppz}
p'_z=\frac{1}{\nu^2}\sum_{j,k=1}^{2^{2q}}(\tilde{b}_j\tilde{b}_k+i\tilde{b}_j\tilde{c}_k-i\tilde{c}_j\tilde{b}_k+\tilde{c}_j\tilde{c}_k)P'_{z,jk},
\end{equation}
where $P'_{z,jk}$ is the $(j,k)$-th entry of $P'_{z}$. The requirement that $\Pi_Q$ never fails indicates that the probability of outputting the POVM elements corresponding to a wrong answer is exactly zero. As indicated by Lemma \ref{pert2}, the approximate distribution $\{p'_z\}$ can be arbitrarily close to the true one (denoted by $\{p_z\}$) when $\bar\epsilon$ is sufficiently small, so the probability corresponding to the wrong answer ${\mathcal{M}_y(x)}$ calculated by Eq. (\ref{ppz}) in $\Pi_C$ is well bounded. Therefore Bob simply sets an appropriate threshold value $\bar\delta(\bar\epsilon)$ that $p'_{\mathcal{M}_y(x)}$ cannot exceed, and refuses to output any $z$ with $p'_z<\bar\delta(\bar\epsilon)$. As long as there exists an answer above this threshold, this protocol is guaranteed to succeed. 

Finally, we determine the appropriate values of $\bar\epsilon$ and $\bar\delta$ in the above protocol $\Pi_C$. To guarantee the existence of at least one valid output, it is sufficient that the upper bound on perturbation on all $p_z$'s, $\delta$, satisfies
\begin{equation}\label{delta}
\delta\equiv \sup_z |p_z-p'_z|< 2^{-m}.
\end{equation}
Then we can simply set the threshold value to
\begin{equation}\label{delta0}
\bar\delta = 2^{-m},
\end{equation}
i.e., Bob only outputs a $z$ with $p'_z\geq 2^{-m}$, which always exists.
By Lemma \ref{pert2}, $\delta<20\bar\epsilon 2^{q}$. Then, according to Eq. (\ref{delta0}), we can set
\begin{equation}\label{eps}
\bar\epsilon = \frac{1}{20}2^{-(m+q)},
\end{equation}
so that $\delta<\bar{\delta}$.
In summary, $\Pi_C$ runs as introduced with $\bar\epsilon$ and $\bar\delta$, respectively, specified by Eqs. (\ref{eps}) and (\ref{delta0}).

By Lemma \ref{cbits}, $C_{CC}(\Pi_C)$ with the above accuracy scales as $O\left((m+q)2^{2q}\right)$. For $m\in \tilde  o(n)$, $m+q\in O(n^\beta)$ always holds, where $0<\beta<1$. Since $2^{2q}\in o(n^\zeta)$ for any positive constant $\zeta$, we simply set $\zeta=1-\beta$, and it can be directly seen that $C_{CC}(\Pi_C)\in o(n^{\beta+\zeta})=o(n)$.
Since $m\in \tilde o(n)$ is within the scope of application of Lemma \ref{classical_lower_bound}, we have reached a contradiction. 
\end{proof}

\section{Detailed Proof of Theorem \ref{error}}\label{appb}

Before presenting the proof, we note that a key point of this theorem is that overhead in communication cost of a successful classical simulation is dependent on the scaling of $(2^{-m}-\gamma)$.
It was shown in an early version of Ref.\ \cite{comm} that only one bit of classical communication is needed for $m\geq\sqrt{n},\gamma=2^{-(m+1)}$. We now sketch the argument here.
Suppose that Alice sends a single bit to Bob indicating whether ${x}$ contains a majority of zeros or a majority of ones. If it is the former case, Bob answers with $\vec{1}\in\{0,1\}^m$ for all $y$, while if it is the latter case, he answers with $\vec{0}\in\{0,1\}^m$.
Without loss of generality, assume that ${x}$ contains a majority of zeros and Bob thus answers with $\vec{1}$. If we denote the number of ones in ${x}$ by $j$, $0\leq j\leq \frac{n}{2}$, the fraction of $y$ for which Bob makes an error, $\mathcal{M}_y\left({x}\right)=\vec{1}$, is given by:
\begin{equation}
\textrm{Probability of error for given }{x}:\left\{
\begin{array}{rl}
\frac{{j \choose m}}{{n \choose m}} & \textrm{for } m\leq j\leq \frac{n}{2},\\
0 & \textrm{for } 0\leq j<m.
\end{array} \right.
\end{equation}
Combining with the fact that the number of ${x}$ with Hamming weight $j$ is ${n \choose j}$, the total probability of error of the strategy, $\epsilon_{t}$, is given by:
\begin{eqnarray}
\epsilon_{t}&=&\frac{\sum_{i=m}^{n/2}{n \choose i}{i \choose m}}{2^{n-1}{n \choose m}}\nonumber\\
&<&\frac{{n \choose \frac{n}{2}}\sum_{i=m}^{n/2}{i \choose m}}{2^{n-1}{n \choose m}}\nonumber\\
&=&\frac{{n \choose \frac{n}{2}}{\frac{n}{2}+1 \choose m+1}}{2^{n-1}{n \choose m}}\nonumber\\
&=&\frac{\frac{n}{2}+1}{m+1}\frac{{n \choose \frac{n}{2}}{\frac{n}{2} \choose m}}{2^{n-1}{n \choose m}}.
\end{eqnarray}
For large $n$ and $m=\sqrt{n}$,
\begin{eqnarray}
\epsilon_{t}&\sim&\frac{1}{2}\sqrt{n}\frac{4^{n/2}}{\sqrt{\frac{\pi n}{2}}}\frac{1}{2^{\sqrt{n}}}\frac{1}{\sqrt{e}}\frac{1}{2^{n-1}}\nonumber\\
&=&\frac{1}{\sqrt{\frac{e\pi}{2}}2^{\sqrt{n}}}\nonumber\\
&<&\frac{1}{2^{\sqrt{n}+1}}.
\end{eqnarray}
Note that in the approximation we used Stirling's approximation for the ${n \choose n/2}$ term, and that
\begin{equation}
\frac{{n/2 \choose m}}{{n \choose m}}\sim\frac{1}{2^m}e^{-1/2}.
\end{equation}
Thus, for $m=\sqrt{n}$, there exists a strategy using one bit of classical communication, when the allowed probability of error is greater than ${1}/{2^{\sqrt{n}+1}}$.
Therefore, it makes sense to pay attention to the regime of even smaller probability of error only, when $m\geq\sqrt{n}$. For this case, the conclusion reduces to a simpler form (Corollary \ref{onebit}).
However for $m<\sqrt{n}$ (where more communication should be needed), it is unsettled whether a non-trivial probability of error can be achieved with constant amount of communication. For now, we conjecture that for $\gamma=2^{-(m+1)}$ and $m\in\Omega({\rm poly}(n))$, $C_{CC}\in O(1)$. However, we have numerical results which indicate that for $m\in o({\rm poly}(n))$, an $O(1)$ size of classical communication cannot guarantee any probability of error that is smaller than $2^{-m}$ in the limit of large $n$. 

The most general form of our rigorous conclusion about the classical simulation when error is allowed goes as follows:

\newtheorem*{thm:error}{Theorem \ref{error}}
\begin{thm:error}[Error-bounded variant of Theorem \ref{main}]
Consider some $h(m)$ such that $\gamma$ satisfies $-\log(2^{-m}-\gamma)\in O(h(m))$. Suppose that for $\EXC_{n,m,\gamma}$ with $\gamma<2^{-m}$, there exists a winning quantum strategy $\Pi^\gamma_Q$ such that $Q_{CC}(\Pi^\gamma_Q)\equiv s\in O(\xi(n))$. Then, one can construct a classical strategy $\Pi^{0^+}_C$ such that $C_{CC}(\Pi^{0^+}_C)=[O(h(m))+O(\xi(n))]2^{O(\xi(n))}$, whose probability of error can be made arbitrarily small. 
\end{thm:error}
\begin{proof}
We revise Bob's local part of the protocol presented in Theorem \ref{main} to devise this $\Pi^{0^+}_C$ as follows. As for the zero error game, given input $x$, Alice prepares a classical message that encodes the real and imaginary parts of all amplitudes of the quantum message $|\psi^\gamma(x)\rangle$ in $\Pi^\gamma_Q$ to accuracy $\bar\epsilon_\gamma$ using $O\left(2^{2s}\log(1/\bar\epsilon_\gamma)\right)$ bits, and sends it to Bob, 
who then normalizes the amplitude vector. 
Instead of classically calculating the probability distribution of the output as in $\Pi_C$, Bob now resorts to local quantum resources. He simply prepares a new quantum state $|\tilde\psi^\gamma(x)\rangle$ according to the normalized state vector (by Lemma \ref{pert}, this state remains close to the original one when $\bar\epsilon_\gamma$ is small), and then feeds it into his original local quantum computation. By Lemma \ref{pert2}, the probability of outputting the wrong answer satisfies:
\begin{equation}\label{p1}
p'_{\mathcal{M}_y( x)}<\gamma+20\bar\epsilon_\gamma 2^{s}.
\end{equation}
As long as $\mathcal{M}_y(x)$ is not the output with the largest probability, i.e., \begin{equation}\label{p2}
p'_{\mathcal{M}_y( x)}<2^{-m},
\end{equation} 
Bob can apply amplitude amplification to suppress the probability of error: he simply repeats his local protocol for $t$ times (he can use the classical message to prepare as many copies of $|\tilde\psi^\gamma(x)\rangle$ as he wants), and outputs the string $z$ that comes out for most times. We denote the probability of error after the whole procedure by $\gamma'$. Then, by the Chernoff bound, for any $\tau>0$, there exists a $\bar t$ such that as long as $t>\bar t$, $\gamma'<\tau$. That is, $\gamma'$ can be made arbitrarily small simply by increasing $t$. Combining Eqs. (\ref{p1}) and (\ref{p2}), we can set
\begin{equation}
\bar\epsilon_\gamma=\frac{2^{-m}-\gamma}{20}2^{-s}
\end{equation}
in the protocol. Since $-\log(2^{-m}-\gamma)\in O(h(m))$ and $s\in O(\xi(n))$, $\log(1/\bar\epsilon_\gamma)\in O(h(m))+O(\xi(n))$. Therefore, $C_{CC}(\Pi_C^{0^+})\in [O(h(m))+O(\xi(n))]2^{O(\xi(n))}$.
Note that the no-cloning theorem is not violated since Bob does not need to copy quantum states, and we do not care about the scaling of $t$ since local computational resource is not limited.
\end{proof}

As argued earlier, by restricting $m\geq\sqrt{n}$, any $\gamma\geq 2^{-(m+1)}$ becomes trivial. Then Theorem \ref{error} takes a simpler form because $\log(2^{-m}-\gamma)\in O(m)$:
\begin{cor}\label{onebit}
Suppose that for $\EXC_{n,m,\gamma}$ with $m\geq\sqrt{n}$ and $\gamma\leq 2^{-(m+1)}$, there exists a winning quantum strategy $\Pi^\gamma_Q$ such that $Q_{CC}(\Pi^\gamma_Q)\equiv s\in O(\xi(n))$. Then, one can construct a classical strategy $\Pi^{0^+}_C$ such that $C_{CC}(\Pi^{0^+}_C)\in[O(m)+O(\xi(n))]2^{O(\xi(n))}$, whose probability of error can be made arbitrarily small.
\end{cor}

\section{Detailed Proof of Theorem \ref{Cerror}}\label{app:Cerror}
Suppose that Bob is allowed to make an error with probability $\gamma$. In other words, for each pair of inputs $({x},y)$, with probability less than or equal to $\gamma$, Bob is allowed to output an $m$-bit string ${z}$ such that ${z}=\mathcal{M}_y\left({x}\right)$. How much classical communication is required from Alice so that Bob does not err with probability more than $\gamma$?
To answer this question, the following definitions and results will be useful. First we formally define the one-way, public-coin randomized communication complexity:
\begin{defn}[One-way, public-coin randomized communication complexity]
For a relation, $f\subseteq\mathcal{X}\times\mathcal{Y}\times\mathcal{Z}$, let $R_{\epsilon}^{1,\text{pub}}\left(f\right)$ denote the communication complexity of the best one-way, public-coin randomized protocol that computes $f$ with error at most $\epsilon$ on all inputs. When referring specifically to the exclusion game, we will replace this by $\bar C_{CC}\left(\text{EXC}_{n,m,\epsilon}\right)$.
\end{defn}

A useful tool for obtaining bounds on the communication complexity is that of rectangle bounds. To define these, we first define (for one-way protocols) \emph{rectangles} and \emph{$\epsilon$-monochromatic functions}.
\begin{defn}[One-way rectangles]
A one-way rectangle, $R$, is defined to be a set ${S}\times\mathcal{Y}$, where $S\subseteq\mathcal{X}$. For a distribution, $\mu$, over $\mathcal{X}\times\mathcal{Y}$, let $\mu_R$ be the distribution formed from $\mu$ by conditioning on $R$. Let $\mu\left(R\right)$ be the probability of the event $R$ under the distribution $\mu$.
\end{defn}
\begin{defn}[One-way $\epsilon$-monochromatic] \label{def:eps mon}
Let $f\subseteq\mathcal{X}\times\mathcal{Y}\times\mathcal{Z}$ be a relation. A distribution, $\lambda$, on $\mathcal{X}\times\mathcal{Y}$ is one-way $\epsilon$-monochromatic for $f$ if there exists a function, $g:\mathcal{Y}\rightarrow\mathcal{Z}$, such that:
\begin{equation}
P_{XY\sim\lambda}\left[\left(X,Y,g(Y)\right)\in f\right]\geq1-\epsilon.
\end{equation}
\end{defn}
With these in place, we now define \emph{rectangle bounds} as follows:
\begin{defn}[Rectangle bound]
Let $f\subseteq\mathcal{X}\times\mathcal{Y}\times\mathcal{Z}$ be a relation. For a distribution, $\mu$, on $\mathcal{X}\times\mathcal{Y}$, the one-way rectangle bound is:
\begin{equation}
\text{rec}_{\epsilon}^{1,\mu}\left(f\right)=\min_R\left\{\log_2\frac{1}{\mu(R)}:R \text{ is one-way rectangle and }\mu_R \text{ is one-way } \epsilon\textrm{-monochromatic.}\right\}.
\end{equation}
The one-way rectangle bound for $f$ is:
\begin{equation}
\text{rec}_{\epsilon}^{1}\left(f\right)=\max_{\mu}\text{rec}_{\epsilon}^{1,\mu}\left(f\right).
\end{equation}
If the above maximization is restricted to product distributions, we can also define:
\begin{equation}
\text{rec}_{\epsilon}^{1,\left[\right]}\left(f\right)=\max_{\mu:\text{product}}\text{rec}_{\epsilon}^{1,\mu}\left(f\right).
\end{equation}
\end{defn}

The utility of rectangle bounds to the problem at hand is given by the following result obtained from \cite{jain2008direct}:
\begin{thm}[\cite{jain2008direct}] \label{Rectangle Bound}
Let $f\subseteq\mathcal{X}\times\mathcal{Y}\times\mathcal{Z}$ be a relation and let $\epsilon\in\left[0,1/6\right]$. Then:
\begin{equation}
R^{1,{\rm pub}}_{\epsilon}\left(f\right)\in\Omega\left({\rm rec}_{\epsilon}^{1,\left[\right]}\left(f\right)\right).
\end{equation}
\end{thm}

This theorem implies the following useful characterization for the communication complexity of the exclusion game for non-zero error, $\gamma$:
\begin{lem} \label{EXC error}
To show a lower bound of $c$ for $\bar C_{CC}\left(\EXC_{n,m,\gamma}\right)$, it is sufficient to show the following. Let $S$ be any subset of $\{0,1\}^n$ of size $2^{n-c}$. Let $A_{M}=\{z(y) \in \{0,1\}^m: y\textrm{  subset of }[n]\textrm{ of size }m\}$ be any set of answers for Bob. Then for at least $\gamma$-fraction of $\{({x},y): {x} \in S, y\textrm{ a subset of }[n]\textrm{ of size }m\}$,  ${z}(y)$ is an incorrect answer for ${x}$.
\end{lem}
\begin{proof}
By Theorem \ref{Rectangle Bound} and the definition of rectangle bounds, we have:
\begin{equation}
\bar C_{CC}\left(\EXC_{n,m,\gamma}\right)\in\Omega\left(\text{rec}^{1,\textit{unif}}_{\gamma}\left(\EXC_{n,m,\gamma}\right)\right),
\end{equation}
where ``unif'' is the product, uniform distribution over $X$ and $Y$. For $R=S\times\mathcal{Y}$:
\begin{equation}
\text{unif}\left(R\right)=\frac{1}{2^c}.
\end{equation}
Thus, if we can not find a set of answers for Bob, $A_{M}$, (in the language of Definition \ref{def:eps mon}, a function $g$) such that $\textit{unif}_R$ is one-way $\epsilon$-monochromatic, then:
\begin{equation}
\text{rec}^{1,\text{unif}}_{\gamma}\left(\EXC_{n,m,\gamma}\right)> c,
\end{equation}
and $\bar C_{CC}\left(\EXC_{n,m,\gamma}\right)\in\Omega(c)$.
\end{proof}


The following fact regarding sums of binomial coefficients will also be used:
\begin{lem} \label{lem: bin sum}
For $m\in\Theta\left(n^{\alpha}\right)$, $1/2 <\alpha <1$:
\begin{equation}
n-\log_2\left(\sum_{i=0}^{m}{n\choose i}\right)\geq n-o(n).
\end{equation}
For $m=\beta n$, $0<\beta<1/2$:
\begin{equation}
n-\log_2\left(\sum_{i=0}^{m}{n\choose i}\right)\in\Omega(n).
\end{equation}
\end{lem}
\begin{proof}
See Appendix C.2 of Ref.\ \cite{comm}.
\end{proof}

Using these lemmas, we can now prove the following result:
\newtheorem*{thm:Cerror}{Theorem \ref{Cerror}}
\begin{thm:Cerror}
For $m\leq\alpha n$ where $0<\alpha<1/2$ is a constant and $ \gamma\leq (n+1)^{-m} $, $\bar C_{CC}(\EXC_{n,m,\gamma})\in\Omega(n) $.
\end{thm:Cerror}
\begin{proof}
First, let $\epsilon={1}/{\left(\sum_{i=0}^{m}{n\choose i}\right)}$ and note that
\begin{equation}
\frac{1}{\sum_{i=0}^{m}{n\choose i}}\geq\frac{1}{\left(n+1\right)^m}.
\end{equation}
Our goal is to determine how large $S$ can be taken to be in Lemma \ref{EXC error} subject to non-zero error $\epsilon$. Note, that from the proof of Theorem 2 in \cite{comm}, we know that, for any choice of $A_M$, at most $\sum_{i=0}^{m-1} {n \choose i}$ strings can be contained in $S$ without introducing any error. An example of when this occurs is when $A_M$ is such that ${z}(y)={0}$ (the $m$-bit string of all zeros) for all $y$ and $S$ consists of all strings with strictly less than $m$ zeros. What strings can be added into this $S$ while keeping the error below $\epsilon$?

There are ${n\choose m}$ strings such that $\mathcal{M}_y\left({x}\right)={0}$ for precisely one value of $y$. These are the strings with precisely $m$ zeros. If we define $S$ as:
\begin{equation}
S=\left\{{x}:{x}\in\{0,1\}^n, \sum_{i=1}^{n} x_i \geq n-m\right\},
\end{equation}
then the fraction of $\{({x},y): {x} \in S, y\textrm{ subset of }[n]\textrm{ of size }m\}$ such that  $z(y)={0}$ is an incorrect answer for ${x}$ is given by:
\begin{equation}
\frac{{n \choose m}}{{n \choose m}\sum_{i=0}^{m}{n \choose i}}=\epsilon.
\end{equation}
As $S$ consists of the maximum number of strings that produce no error and strings that produce only one error, it is clear that this is the largest $S$ can be taken to be for error given by $\epsilon$. Thus, by Lemma \ref{EXC error}:
\begin{align}
\bar C_{CC}\left(\EXC_{n,m,\epsilon}\right)&\in\Omega\left(n-\log_2\left|S\right|\right),\nonumber\\
&=\Omega\left(n-\log_2\left(\sum_{i=0}^{m} {n \choose i}\right)\right).
\end{align}
By Lemma \ref{lem: bin sum}, for $m\in\Theta(n^\alpha), 1/2<\alpha<1$, we obtain: 
\begin{equation}
\bar C_{CC}\left(\EXC_{n,m,\epsilon}\right)\in\Omega\left(n\right).
\end{equation}
Finally, as $\epsilon\geq\left(n+1\right)^{-m}$, the scaling holds for error parametrized by $\gamma$ as given in the statement of the theorem.
\end{proof}

\section{Detailed Proof of Theorem \ref{compression}}\label{appd}
In the PJO quantum strategy \cite{comm}, upon receiving ${x}$, Alice sends the state
\begin{align}
\ket{\Phi(x)}&=\bigotimes_{i=1}^{n}\left[ \cos\left(\frac{\theta_m}{2}\right)\ket{0}+\left(-1\right)^{x_i} \sin\left(\frac{\theta_m}{2}\right)\ket{1}\right]\nonumber\\
&=\sum_{{r}\in\{0,1\}^n} \left(-1\right)^{{x}\cdot{r}} \left[\cos\left(\frac{\theta_m}{2}\right)\right]^{n-|{r}|}\left[\sin\left(\frac{\theta_m}{2}\right)\right]^{|{r}|} \ket{{r}} \label{Eq:Full State}
\end{align}
where $\theta_m=2\tan^{-1}\left(2^{1/m}-1\right)$.

Suppose that instead of directly sending $\ket{\Phi(x)}$, Alice compresses the message by projecting the state onto the space spanned by the computational basis vectors with with Hamming weight (the number of ones) at most $k$. The compressed quantum message reads as
\begin{equation}
\ket{\Phi^{(k)}(x)}=\frac{1}{\sqrt{A_k}}\sum_{\substack{{r}\in\{0,1\}^n\\|{r}|\leq k}} \left(-1\right)^{{x}\cdot{r}} \left[\cos\left(\frac{\theta_m}{2}\right)\right]^{n-|{r}|}\left[\sin\left(\frac{\theta_m}{2}\right)\right]^{|{r}|} \ket{{r}}, \label{eq:Compressed State1}
\end{equation}
where
\begin{equation}
A_k=\sum_{i=0}^{k}{n \choose i}\left[\cos\left(\frac{\theta_m}{2}\right)\right]^{2\left(n-i\right)}\left[\sin\left(\frac{\theta_m}{2}\right)\right]^{2i}.\end{equation}
This compression reduces the number of qubits Alice sends to $\log\left(\sum_{i=0}^{k} {n \choose i}\right)$. Assuming that Bob performs the same measurement on the qubits specified by $y$ as he would without the compression:
\begin{equation}
\ket{\zeta(z)}=\frac{1}{\sqrt{2^m}}\left(\ket{{0}}-\sum_{{s}\neq{0}} \left(-1\right)^{z\cdot{s}}\ket{{s}}\right), 
\end{equation}
this would lead to some probability of error, $\epsilon_k$. If $\rho^{k}_{{x},y}=\textrm{Tr}_{\backslash y}\left[\ketbra{\Phi^{(k)}(x)}{\Phi^{(k)}(x)}\right]$ denotes the state sent by Alice restricted to the locations specified by $y$, then
\begin{equation}\label{ek}
\epsilon_k=\bra{\zeta(\mathcal{M}_y(x))}\rho^{k}_{{x},y}\ket{\zeta(\mathcal{M}_y(x))}.
\end{equation}
To bound $\epsilon_k$, we make use of the following lemma:
\begin{lem} \label{lemma:trace distance}
For $\ket{\Phi(x)}$, $\ket{\Phi^{(k)}(x)}$ and $\epsilon_{k}$ respectively defined in Eqs. (\ref{Eq:Full State}), (\ref{eq:Compressed State1}) and (\ref{ek}):
\begin{equation}
\sqrt{1-\left|\bracket{\Phi(x)}{\Phi^{(k)}(x)}\right|^2}\geq \epsilon_{k}.
\end{equation}
Note that $\bracket{\Phi(x)}{\Phi^{(k)}(x)}$ is independent of ${x}$.
\end{lem}
\begin{proof}
Recall that the trace distance between two density matrices, $\rho$ and $\sigma$, is given by
\begin{equation}
D\left(\rho,\sigma\right)=\frac{1}{2}\textrm{Tr}\left[\sqrt{\left(\rho-\sigma\right)^{\dag}\left(\rho-\sigma\right)}\right].
\end{equation}
For pure states, $\ket{\psi}$ and $\ket{\phi}$, this reduces to
\begin{equation}
D\left(\ket{\psi},\ket{\phi}\right)=\sqrt{1-\left|\bracket{\psi}{\phi}\right|^2}.
\end{equation}
We will also need the following facts. Firstly, as the trace distance never increases under local operations, for bipartite states, $\rho_{AB}$ and $\sigma_{AB}$:
\begin{equation}
D\left(\rho_{AB},\sigma_{AB}\right)\geq D\left(\rho_{A},\sigma_{A}\right).
\end{equation}
Second, by Eq.\ (9.22) in \cite{nc}:
\begin{equation}
D\left(\rho,\sigma\right)=\max_P \textrm{Tr}\left[P\left(\rho-\sigma\right)\right],
\end{equation}
where the maximization is taken over all projectors $P$.
Combining these facts, we obtain:
\begin{align}
\epsilon_k&=\bra{\zeta(\mathcal{M}_y(x))}\rho^{k}_{{x},y}\ket{\zeta(\mathcal{M}_y(x))}\nonumber\\
&=\bra{\zeta(\mathcal{M}_y(x))}\rho^{k}_{{x},y}\ket{\zeta(\mathcal{M}_y(x))}
-\bra{\zeta(\mathcal{M}_y(x))}\rho^{n}_{{x},y}\ket{\zeta(\mathcal{M}_y(x))}\nonumber\\
&\leq D\left(\rho^{k}_{{x},y},\rho^{n}_{{x},y}\right)\nonumber\\
&\leq D\left(\ket{\Phi^{(k)}(x)},\ket{\Phi(x)}\right)\nonumber\\
&=\sqrt{1-\left|\bracket{\Phi(x)}{\Phi^{(k)}(x)}\right|^2},
\end{align}
as required.
\end{proof}

Lemma \ref{lemma:trace distance} enables us to prove the following theorem:

\newtheorem*{thm:QCC}{Theorem \ref{compression}}
\begin{thm:QCC}
For $ m\in \Theta(n^\alpha) $, $ 1/2<\alpha<1 $ and $ \gamma\geq(n+1)^{-m}$, $\bar Q_{CC}(\EXC_{n,m,\gamma})\in O(m^{1+\delta}) $ for any $\delta>0 $.
\end{thm:QCC}
\begin{proof}
\begin{align}
\sqrt{1-\left|\bracket{\Phi(x)}{\Phi^{(k)}(x)}\right|^2}&=\sqrt{1-\sum_{i=0}^{k}{n \choose i}\left[\cos\left(\frac{\theta_m}{2}\right)\right]^{2n-2i}\left[\sin\left(\frac{\theta_m}{2}\right)\right]^{2i}}\nonumber\\
&=\sqrt{\sum_{i=k+1}^{n}{n \choose i}\left[\cos\left(\frac{\theta_m}{2}\right)\right]^{2n-2i}\left[\sin\left(\frac{\theta_m}{2}\right)\right]^{2i}}.
\end{align}
Now,
\begin{align}
{n \choose i}&\leq \left(\frac{ne}{i}\right)^i,\\
\cos^2\left(\frac{\theta_m}{2}\right)&\leq 1, \\
\sin^2\left(\frac{\theta_m}{2}\right)&<\frac{1}{m^2}, \quad \textrm{for large }m;
\end{align}
so, for large $m$,
\begin{align}
1-\left|\bracket{\Phi(x)}{\Phi^{(k)}(x)}\right|^2&<\sum_{i=k+1}^{n}\left(\frac{ne}{i}\right)^i\left(\frac{1}{m}\right)^{2i}\nonumber\\
&\leq \left(n+1\right)\left(\frac{ne}{m^2k}\right)^k,
\end{align}
as the $i=k+1$ term decays slowest for $m\in\omega\left(\sqrt{n}\right)$.
For this bound to be less than $\gamma^2=\left(n+1\right)^{-2m}$, we require:
\begin{align}
\left(\frac{m^2 k}{ne}\right)^k &> \left(n+1\right)^{2m+1}\\
k \log\left(\frac{m^2 k}{ne}\right) &> \left(2m+1\right)\log\left(n+1\right).
\end{align}
To satisfy this asymptotically, it suffices to take $k=m^{1+\eta}$ with any $\eta>0$. The number of qubits sent (which, by Lemma \ref{lemma:trace distance}, achieves a probability of error $\leq\left(n+1\right)^{-m}$) is then:
\begin{align}
\log\left(\sum_{i=0}^{m^{1+\beta}}{n \choose i}\right) &\leq \log\left(\left(n+1\right)^{m^{1+\beta}}\right)\nonumber\\
&=m^{1+\beta}\log\left(n+1\right).
\end{align}
This can scale as $O(m^{1+\delta})$ for any $\delta>0$, by choosing some $\eta>\delta$. Thus, $\bar Q_{CC}\left(\text{EXC}_{n,m,\gamma}\right)\in O\left(m^{1+\delta}\right)$ for any $\delta>0$.
\end{proof}

\end{document}